\documentclass[
final, nomarks
]{dmtcs-episciences}

\usepackage[utf8]{inputenc}
\usepackage{subfigure}

\usepackage[round]{natbib}

\usepackage{graphicx}
\usepackage{amsmath, amssymb, amsfonts, amstext, mathdots, latexsym}
\usepackage{algorithm, algorithmic}
\usepackage{caption, subcaption}
\usepackage{multirow, slashbox, rotating}
\usepackage{tipa, arcs, fix-cm}
\usepackage{enumerate, setspace, makeidx, epsfig, verbatim}

\newcommand{\ignore}[1]{}

\newcommand{\norm}[1]{\left\lVert#1\right\rVert}

\newtheorem{theorem}{Theorem}[section]

\newtheorem{proposition}[theorem]{Proposition}
\newtheorem{lemma}[theorem]{Lemma}
\newtheorem{corollary}[theorem]{Corollary}

\newtheorem{observation}[theorem]{Observation}

\newtheorem{definition}[theorem]{Definition}

\author[Konstantinos Georgiou et al.]{Konstantinos Georgiou\affiliationmark{1}
  \and Somnath Kundu\affiliationmark{1}
  \and Pawe\l{} Pra\l{}at\affiliationmark{1}}
\title[The Fagnano Triangle Patrolling Problem]{The Fagnano Triangle Patrolling Problem\thanks{Research supported in part by NSERC.}
\thanks{
This is the full version of the paper by~\cite{KSP23-fagnano} with the same title which appeared in the proceedings of the 25th International Symposium on Stabilization, Safety, and Security of Distributed Systems  (SSS’23), October 2-4, 2023, in New Jersey, USA.
}
}
\affiliation{
  Toronto Metropolitan University, Toronto, Canada}

\begin{document}
\publicationdata
{vol. 27:3}
{2025}
{4}
{10.46298/dmtcs.13421}
{2024-04-17; 2024-04-17; 2025-02-10; 2025-07-01}
{2025-07-18}

\maketitle
\begin{abstract}
We investigate a combinatorial optimization problem that involves patrolling the edges of an acute triangle using a unit-speed agent. 
The goal is to minimize the maximum (1-gap) idle time of any edge, which is defined as the time gap between consecutive visits to that edge. 
This problem has roots in a centuries-old optimization problem posed by Fagnano in 1775, who sought to determine the inscribed triangle of an acute triangle with the minimum perimeter. 
It is well-known that the orthic triangle, giving rise to a periodic and cyclic trajectory obeying the laws of geometric optics, is the optimal solution to Fagnano's problem. 
Such trajectories are known as Fagnano orbits, or more generally as billiard trajectories. 
We demonstrate that the orthic triangle is also an optimal solution to the patrolling problem.

Our main contributions pertain to new connections between billiard trajectories and optimal patrolling schedules in combinatorial optimization. 
In particular, as an artifact of our arguments towards proving optimality of our results, we introduce a novel 2-gap patrolling problem that seeks to minimize the visitation time of objects every three visits. 
We prove that there exist infinitely many well-structured billiard-type optimal trajectories for this problem, including the orthic trajectory, which has the special property of minimizing the visitation time gap between any two consecutively visited edges. Complementary to that, we also examine the cost of dynamic, sub-optimal trajectories to the 1-gap patrolling optimization problem. These trajectories result from a greedy algorithm and can be implemented by a  computationally primitive mobile agent.
\end{abstract}

\section{Introduction}
\label{sec: intro}

Patrolling refers to the perpetual monitoring, protection, and supervision of a domain or its perimeter using mobile agents. 
In a typical patrolling problem involving one mobile agent, the agent must move through a given domain in order to monitor or check specific locations or objects. The objective is to find a trajectory that satisfies certain constraints and/or that addresses quantitative objectives, such as minimizing the total distance traveled or maximizing the frequency of visits to certain areas.
The purpose of patrolling could be to detect any intrusion attempts, monitor for possible faults or to identify and rescue individuals or objects in a disaster environment,  
and for this reason, such problems arise in a variety of real-world applications, such as security patrol routes, autonomous robot navigation, and wildlife monitoring.
Overall the subject of patrolling has seen a growing number of applications in Computer Science, including Infrastructure Security, Computer Games, perpetual domain-surveying, and monitoring in 1D and 2D geometric domains.

In addition to its practical applications, patrolling has emerged (not as a combinatorial  optimization problem) in the context of theoretical physics. In particular, the problem of finding periodic trajectories in billiard systems has been a topic of interest for many years. A billiard system is a model of a particle or a waveform moving inside a domain (typically polygonal, but also elliptical, convex, or even non-convex region) and reflecting off its boundaries according to the laws of elastic collision which naturally translate to geometric conditions. For this reasonthe problem of finding periodic trajectories in a billiard system is equivalent to finding a closed path in the domain that satisfies these geometric conditions. 

One important example of a periodic trajectory in billiard systems is the so-called Fagnano orbit on acute triangles, a periodic, closed (and piece-wise linear) curve that visits the three edges of an acute triangle. Fagnano orbits, named after the Italian mathematician Giulio Fagnano who first studied them in the mid-18th century, arise as solutions to the optimization problem which asks for the shortest such curve. 
In this work we explore further connections between billiard trajectories and patrolling as a combinatorial optimization problem. 
In particular, we are asking what are the patrolling strategies for the edges of an acute triangle that optimize standard frequency-related objectives. 
Our findings demonstrate that a family of Fagnano orbits are optimal solutions to the corresponding combinatorial optimization problems. This enriches the connection between combinatorial patrolling and billiard trajectories, a relationship that has also been explored in prior work, e.g., in~\citep{Damaschke21}, where certain discrete and combinatorial patrolling strategies in a transformed discrete configuration space 
exhibit reflective path patterns that resemble billiard dynamics. 
Our result further supports this connection in continuous two-dimensional geometric domains.

Billiard trajectories are characterized by piecewise linear motion, obeying the laws of geometric optics upon reaching the boundaries of the domain. The fact that trajectories remain linear, except at the boundaries, is due to the uniform speed in the medium. Similarly, straight-line motion ensures minimal travel time, as speed is assumed to be constant. The reflection rule, which follows the laws of geometric optics, is a local property that applies only when the trajectory reaches a boundary. This guarantees that at each reflection point, and locally, movement occurs in the shortest possible time.
As a result, a closed billiard trajectory corresponds to the shortest travel path that visits a specific sequence of reflection points. However, this does not directly imply that the same billiard trajectory is optimal for patrolling the associated boundary segments. This is because one must first show that the optimal patrolling strategy must be periodic and cyclic.
In this regard, our contribution clarifies the connection (at least for billiard trajectories in triangles) by introducing a new patrolling problem in which efficiency is measured not by consecutive segment visits but by the time gaps between every three visits.

\section{Related Work}
\label{sec: related work}

Patrolling problems are a fundamental class of problems in computational geometry, combinatorial optimization, and robotics that have attracted significant research interest in recent years.
Due to their practical applications, they have received extensive treatment in the realm of robotics, see for example~\citep{ARSTMCC04,C04,EAK09,ESK08,HK08,MRZD02,YWB03}, as well as surveys~\citep{basilico2022recent,HuangZHH19,PortugalR11}. 
When patrolling is seen as part of infrastructure security, it leads to a number of optimization problems~\citep{KranakisK15}, with one particular example being the identification of network failures or web pages requiring indexing~\citep{MRZD02}. 

Combinatorial trade-offs of triangle edge visitation costs have been explored by~\cite{GeorgiouKP21}. 
In contrast, the current work pertains to the cost associated with the perpetual monitoring of the triangle edges by a single unit speed agent. 
Numerous variations of similar patrolling problems have been explored in computational geometry, which vary depending on the application domain, patrolling specifications, agent restrictions, and computational abilities.
Many efficient algorithms have been developed for several of these variants, utilizing a range of techniques from graph theory, computational geometry, and optimization, see survey by \cite{czyzowicz2019patrolling} for some recent developments. 
Some examples of studied domains include the bounded line segment in~\citep{KawamuraK15},
networks in~\citep{YWB03},
polygonal regions in~\citep{TanJ14}, trees  in~\citep{CzyzowiczKKT16},
disconnected fragments of one dimensional curves in~\citep{CollinsCGKKKMP13},
arbitrary polygonal environments in~\citep{PasqualettiFB10} (with a reduction to graphs), 
or even 3-dimensional environments in~\citep{FredaGPGDSC19}.

Identifying optimal patrolling strategies can be computationally hard as shown by~\cite{Damaschke20}, 
while even in seemingly easy setups the optimal trajectories can be counter-intuitive, see~\citep{kawamura2020simple}. 
The addition of combinatorial specifications has given rise to multiple intriguing variations, including 
the requirement of uneven coverage by~\cite{ChuangpishitCGG18,PiciarelliF20} or waiting times by~\cite{Damaschke21},
the presence of high-priority segments by~\cite{Ponce20}, 
and patrolling with distinct speed agents by~\cite{CzyzowiczGKMP16a}. 
Patrolling has also been studied extensively, e.g. by~\cite{LeeFeketeM16}, from the perspective of distributed computing, 
while the class of these problems also admit a game-theoretic interpretation between an intruder and a surveillance agent, see~\citep{AlpernMP11,Garrec19}.

Maybe not surprisingly, the optimal patrolling trajectories that we derive are in fact billiard-type trajectories, 
that is, periodic and cyclic trajectories obeying the standard law of geometric optics,
and which are referred to as Fagnano orbits specifically when the underlying billiard/domain is triangular. 
Fagnano orbits have been studied extensively both experimentally by~\cite{lafargue2014localized} and theoretically by~\cite{Troubetzkoy09}.
Billiard-type trajectories have been explored in equilateral triangles by~\cite{baxter2008periodic},
obtuse triangles by~\cite{Halbeisen-Hungerbuhler00}, as well as polygons by~\cite{vorobets1992periodic}.
More recently, there have been studies on ellipses by~\cite{Garcia19} and general convex bodies by~\cite{karasev2009periodic}, 
or even fractals by~\cite{LapidusN2011arxiv} and polyhedra by~\cite{bedaride2011periodic}, with the list of domains or trajectory specifications still growing. 

\section{Main Definitions and Results}
\label{sec: main definitions and results}

A patrolling schedule $S$ (or simply a schedule) for triangle $\Delta$ with edges (line segments) $E=\{\alpha,\beta,\gamma\}$ is an infinite sequence $\{s_i\}_{i\geq0}$, where each $s_i$ is a point on a line segment of $E$ that we also denote by $e(s_i)$, i.e. $e(s_i)\in E$ for each $i\geq 0$. For notational convenience, we will also use symbols the $\alpha,\beta,\gamma$ to denote the lengths of the corresponding edges. 
Similarly, $s_is_j$ will denote the segment with endpoints $s_i,s_j$ as well as its length. 
When $e(s_i)=\delta\in E$ we say that segment $\delta$ and point $s_i$ are visited at step $i$ of the schedule.
We will only be studying \textit{feasible schedules}, i.e. schedules for which eventually all segments in $E$ are visited and infinitely often. That is, we require for all $\delta \in E$, and all $k\in \integers$, that there exists $l\geq k$ with $e(s_l)=\delta$.

For simplicity, our notation above is tailored to points $s_i$ that are not vertices of $\Delta$. When $s_i$ is a vertex of $\Delta$ we assume that both incident edges are visited.
We also think of schedule $S$ as the trajectory of a unit speed agent, and hence we refer to the time between the visitation of $s_j,s_{j+\ell}$ as the sum of the lengths of segments $s_{j+i}s_{j+i+1}$ 
over $i \in \{0,\ldots, \ell-1\}$. 

A schedule $S$ is called: \\
- \textit{cyclic} if $\{e(s_0), e(s_1), e(s_2)\}=E$ and $e(s_{i+3})=e(s_{i})$, for every $i\geq0$, and \\ 
- \textit{$k$-periodic} (for $k\geq 3$) if $s_{i+k}=s_{i}$,
for every $i\geq 0$. 

For any segment $\delta \in E$ we define its \textit{$t$-gap sequence}, $g^t(\delta)$, that records the visitation time gaps of $\delta$ over every $t+1$ consecutive visitations. In particular, $t=1$ corresponds to the standard \emph{idle time} considered previously, and that measures the additional time it takes for each object to be revisited, after each visitation. 
Formally, 
let $e(s_j)=e(s_{j'})=\delta$ and
suppose that points $s_j,s_{j'}$ are the $k$-th and $(k+t)$-th visitation of $\delta$, respectively. Then the time between the visitations of $s_j,s_{j'}$ is exactly the value of $k$-th element of sequence $g^t(\delta)$. From this definition, it is also immediate that 
$\left(g^t(\delta)\right)_i =\sum_{j=i}^{i+t-1}\left(g^1(\delta)\right)_j$. 

The \textit{$t$-gap} $G^t(\delta)$ of $\delta \in E$ is defined as 
$\sup_i \left(g^t(\delta)\right)_i$, 
while the \textit{$t$-gap} $G^t$ of schedule $S$ for edges $E$ (hence for input triangle $\Delta$) is defined as $\max_{\delta \in E}G^t(\delta)$.
When it is clear from the context, we will abbreviate $G^1$ simply by $G$.

\subsection{Main Contributions and More Terminology}
\label{sec: main contributions}

In this section we summarize our main contributions, pertaining to optimal $1$-gap and $2$-gap patrolling schedules of acute triangles. 
We begin with a self-contained proof of the optimality of $1$-gap patrolling schedules, restricted to cyclic and $3$-periodic schedules.  
We present this proof as a warm-up and a reference for the proof of Lemma~\ref{lem: cost of orthic}, which provides a closed formula for the 1-gap of the optimal cyclic and 3-periodic schedule. 
To present our result, we recall the so-called \textit{orthic triangle}, a pedal-type triangle associated with an acute triangle $\Delta$. This triangle is inscribed in $\Delta$, with its vertices given by the projections of $\Delta$'s orthocenter (the intersection of its altitudes) onto its three edges.  
Note also that any $3$-periodic cyclic schedule corresponds to a triangle inscribed in $\Delta$. The following theorem, attributed to Fagnano in 1775, is proved in Section~\ref{sec: optimal cyclic 1-periodic}, where we also introduce key concepts that support our main contributions.

\begin{theorem}[Fagnano's Theorem]
\label{thm: optimal 1-gap cyclic 3-periodic}
The optimal $1$-gap $3$-periodic cyclic patrolling schedule of a triangle $\Delta$ is its orthic triangle.  
\end{theorem}

Towards our goal to derive optimal $1$-gap schedules, we find all (infinitely many) optimal $2$-gap cyclic schedules, which are in fact billiard-type trajectories.
We prove the next theorem in Section~\ref{sec: optimal 2-gap}.
\begin{theorem}
\label{thm: optimal 2-gap cyclic}
There are infinitely many optimal $2$-gap cyclic schedules of a triangle $\Delta$, that include also the orthic triangle. Every $2$-gap optimal schedule is $6$-periodic and has value equal to 2 times the perimeter of the orthic triangle. Moreover, each optimal schedule is made up of segments that are parallel to the edges of the orthic triangle. 
\end{theorem}

Then in Section~\ref{sec: optimal cyclic} we derive our main contribution.
\begin{theorem}
\label{thm: optimal 1-gap cyclic}
The orthic triangle of a triangle is an optimal $1$-gap patrolling schedule.
\end{theorem}
In the same section, we quantify the $1$-gap of the orthic triangle and compare it to the optimal $2$-gap schedules. Specifically, we examine which of the optimal $2$-gap schedules minimizes the maximum time between visits to any two distinct edges of $\Delta$. We then prove that, in this multi-objective optimization problem, the orthic schedule remains optimal.

From our previous contributions, we conclude that a mobile agent whose task is to $1$-gap optimally patrol a triangle $\Delta$ needs the capability to compute the base points of $\Delta$'s altitudes. This naturally raises the question: can we obtain efficient solutions using a primitive agent?
In this work, we explore the primitive paradigm of a greedy-type algorithm, which constructs a solution by iteratively making a sequence of irrevocable choices, each locally optimal according to a given heuristic. In our case, the heuristic is chosen to greedily minimize the time until the next object visitation.
A fundamental algorithmic question is whether such simple heuristics can yield optimal or near-optimal solutions. We quantify the deviation of the resulting solution from optimality and, perhaps more importantly, demonstrate that the trajectory eventually converges to a periodic and cyclic pattern.
Indeed, in Section~\ref{sec: greedy}, we formalize this result and provide the necessary technical details.

\begin{theorem}
\label{thm: greedy}
There is a greedy-type schedule that converges to a~ $3$-periodic cyclic schedule whose $1$-gap is off from the $1$-gap optimal cyclic schedule by a factor 
$\gamma \in [1,\gamma_0]$, where $\gamma_0=\sqrt{2}/2+1/2$, and 
$\gamma$ admits a closed formula as  a function of the angles of the given triangle. 
\end{theorem}

The greedy schedule we will analyze will be given by a sequence of points on the triangle's edges, which will visit these edges cyclically. Therefore, convergence to the periodic schedule refers to the piecewise convergence of these points, on each edge, to a specific points with respect to the $\ell_1$ metric.
It will follow from our analysis that our greedy algorithm will be nearly optimal for any acute triangle with one arbitrarily small angle, and it will be the worst off from the optimal solution when the given triangle is a right isosceles.

\section{The $1$-Gap Optimal $3$-Periodic Cyclic Schedule}
\label{sec: optimal cyclic 1-periodic}

There are many proofs known for the fact that inscribed triangle with the shortest perimeter is its orthic triangle. In the language of triangle patrolling, the statement is equivalent to that that the optimal $1$-gap $3$-periodic cyclic schedule of a triangle is its orthic triangle, articulated in Theorem~\ref{thm: optimal 1-gap cyclic 3-periodic}. For completeness, we provide next a self-contained proof.

\begin{proof}[of Theorem~\ref{thm: optimal 1-gap cyclic 3-periodic}] 
We consider triangle $ABC$, as in Figure~\ref{fig: orthic is opt}. First we find the inscribed triangle of minimum perimeter, and then we show the optimizer is the orthic triangle. 

 \begin{figure}[h!]
\centering
  \includegraphics[width=9cm]{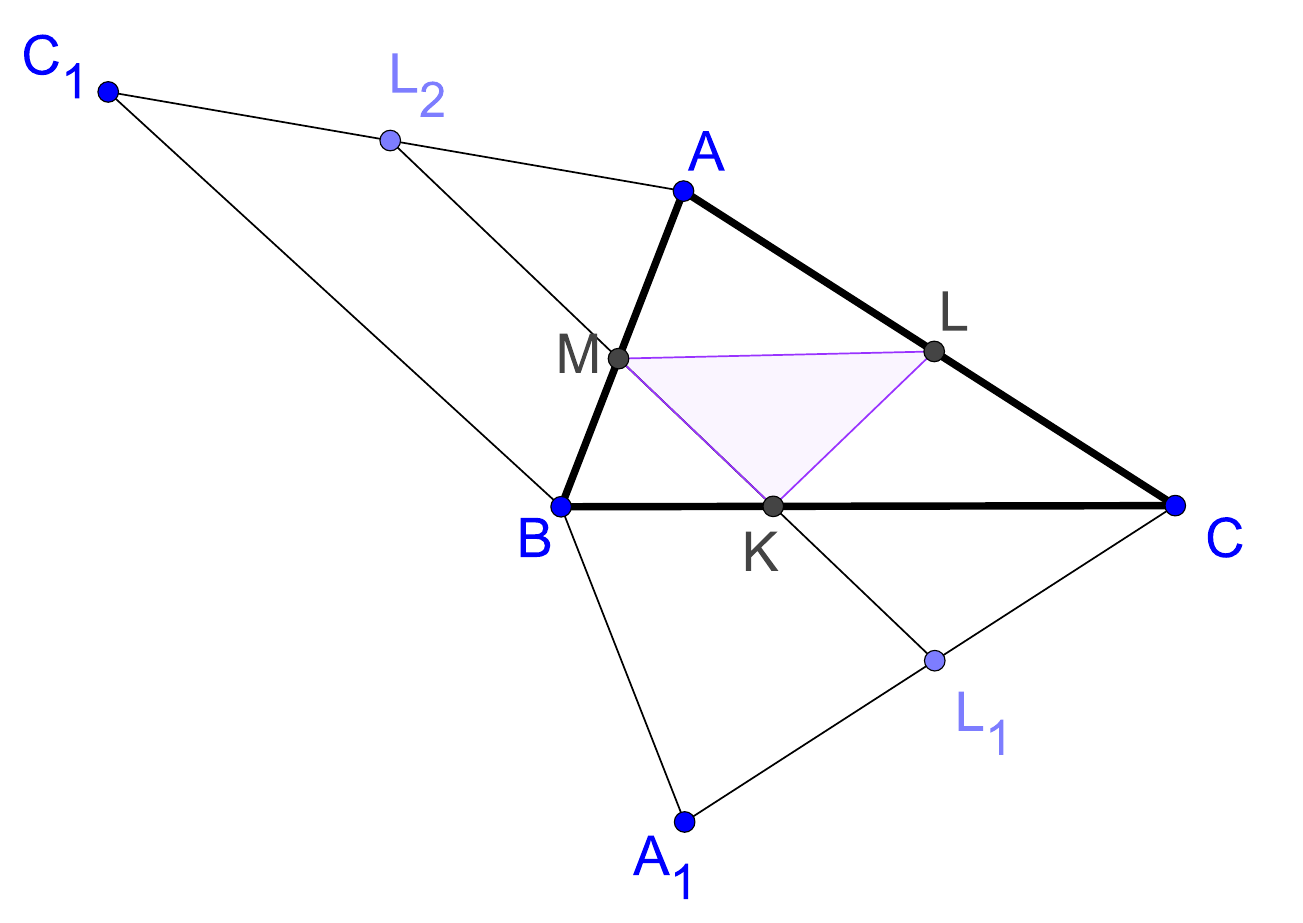}
\caption{Triangle figure supporting the proof of Fagnano's Theorem.}
\label{fig: orthic is opt}
\end{figure}

We start with an arbitrary point $L$ on $AC$, and we find the optimal points $K,M$ on $BC,AB$, respectively, so as to minimize the perimeter of $KLM$ as a function of $L$. Then, we show how to choose $L$ so as to minimize the perimeter.

Consider the reflection $A_1$ of $A$ about $BC$, the reflection $C_1$ of $C$ about $AB$, and the reflections $L_1,L_2$ of $L$ about $BC,AB$, respectively. Consider also the intersections $K,M$ of $L_1L_2$ with $BC,AB$, respectively. Clearly, the optimal way to start from $L$, visit edge $BC$, then $AB$ and then return to $L$ is by following the edges of triangle $KLM$. Moreover, the perimeter of $KLM$ equals the length of segment $L_1L_2$. Next we minimize the length of $L_1L_2$ as a function of $L$. 

In this direction, we consider a Cartesian system centered at $B$ where $A=(p,q), B=(0,0), C=(1,0)$ (in particular, w.l.o.g.\ we assume that $BC=1$). Note that $A_1=(p,-q), C_1=(\cos(2B), \sin(2B))$, where also 
\begin{equation}
\label{equa: A coordinates}
p=\frac{\cos(B)\sin(C)}{\sin(B+C)}, 
~~
q=\frac{\sin(B)\sin(C)}{\sin(B+C)}.
\end{equation}

Point $L$ is a convex combination of $A,C$, hence there exists $x\in [0,1]$ such that $L=xC+(1-x)A$. Therefore, 
$
L_1 = xC+(1-x)A_1, ~~L_2=xC_1+(1-x)A,
$
so that
\begin{align*}
\norm{L_1-L_2}^2
& = \norm{
x(C-C_1) + (1-x)(A_1-A)
}^2 
\\
& =
\norm{
x
\left(
\begin{array}{c}
\cos(2B)-1\\
\sin(2B)
\end{array}
\right)
+
(1-x)
\left(
\begin{array}{l}
0\\
2q
\end{array}
\right)
}^2. 
\end{align*}

It follows that $\norm{L_1-L_2}^2$ is convex in $x$ (degree 2 polynomial), and elementary calculations show that the minimum is attained at 
$
x_0
=\frac{\cos(A)\sin(C)}{\sin(B)}.
$
It can be seen that for all acute triangles, we have $x_0\in [0,1]$. Then, the minimum patrolling trajectory has length 
$$
\norm{
x_0(C-C_1) + (1-x_0)(A_1-A)
}
=
2 \sin(B) \sin(C).
$$

Note that the choice of $x_0$ determines all points $K,L,M$. Now we show that $KLM$ is the orthic triangle. 
In order to show that $K$ is the base of the altitude corresponding to $A$, we verify that points $(p,0), L_1, L_2$ are collinear (and hence $K=(p,0)$). For this observe that $L_1,L_2$ are already expressed as a function of $p,q,x_0$, and hence the claim follows by straightforward calculations. 

Next we show that $L$ is the base of the altitude corresponding to $B$. For this we verify that $KL$ is perpendicular to $AC$. Indeed, 
$
L = x_0C+(1-x_0)A, ~~L_2=x_0C_1+(1-x_0)A = (x_0 +(1-x_0)p,(1-x_0)q),
$
while vector $AC$ is $(1-p,-q)$. Taking the inner product of the vectors gives
$(x_0 +(1-x_0)p)(1-p)-(1-x_0)q^2$ which, after elementary calculations reduces to $0$, as promised. 

Finally, we verify that $M$ is the base of the altitude corresponding to $C$. For this, we compute the projection of $C=(1,0)$ onto the line passing through $A,B$ which reads as $py-qx=0$, which is point $\frac{p}{p^2+q^2}(p,q)$. Finally, elementary calculations can verify that the latter point, together with $K,L_2$ are collinear, and hence $M=\frac{p}{p^2+q^2}(p,q)$ as promised. 
 \end{proof}

The next complementary lemma effectively provides a formula for the optimal $1$-gap of cyclic 3-periodic schedules.

\begin{lemma}
\label{lem: cost of orthic}
Let $p$ be the perimeter of an acute triangle. Then, the perimeter of its orthic triangle is given by 
\begin{equation}
\label{equa: orthic cost}
2p
\left(
\frac1{\sin(B)\sin(C)}
+
\frac1{\sin(A)\sin(C)}
+
\frac1{\sin(A)\sin(B)}
\right)^{-1}.
\end{equation}

\end{lemma}

\begin{proof}

As in the proof of Theorem~\ref{thm: optimal 1-gap cyclic 3-periodic} we assume that $A=(p,q), B=(0,0), C=(1,0)$ and hence that $\alpha=1$. From the derived formulas for the coordinates of points $K,L,M$ we have that 
\begin{align*}
\norm{K-L}&=\frac12 \csc(B+C)\sin(2C) \\
\norm{K-M}&=\frac12 \csc(B+C)\sin(2B) \\
\norm{M-L}&=- \cos(B+C).
\end{align*}
But then, elementary trigonometric calculations give
$$
\norm{K-L}+
\norm{K-M}+
\norm{M-L}=2 \sin(B)\sin(C).
$$
It follows that for arbitrary edge length $\alpha$ (not necessarily equal to 1), we have that the perimeter of the orthic triangle equals $X=2 \alpha \sin(B)\sin(C)$. Due to the symmetry of the formula, the perimeter must be also equal to $2\beta \sin(A)\sin(C)
$ and to $2\gamma \sin(A)\sin(B)$. 
We conclude that 
$$
\alpha = \frac{X}{2\sin(B)\sin(C)}, ~~
\beta = \frac{X}{2\sin(A)\sin(C)}, ~~
\gamma = \frac{X}{2\sin(A)\sin(B)}.
$$
So if we denote by $p$ the perimeter of the given triangle, i.e. $p=\alpha+\beta+\gamma$, adding the previous equations and solving for $X$ gives the promised formula. 
\ignore{
\SK{
\\\\
The Expression can be further simplified as follows.
\begin{align*}
&
\frac1{\sin(B)\sin(C)}
+
\frac1{\sin(A)\sin(C)}
+
\frac1{\sin(A)\sin(B)}\\
&=\frac{\sin(A) + \sin(B) + \sin(C)}{\sin(A)\sin(B)\sin(C)}\\
&=\frac{4\cos\frac{A}{2}\cos\frac{B}{2}\cos\frac{C}{2}}{\left(2\sin\frac{A}{2}\cos\frac{A}{2}\right) \cdot \left(2\sin\frac{B}{2}\cos\frac{B}{2}\right)\left(2\sin\frac{C}{2}\cos\frac{C}{2}\right)}\\
&=\frac{1}{2\cdot \sin\frac{A}{2}\cdot \sin\frac{B}{2}\cdot \sin\frac{C}{2}}
\end{align*}
}
}
\ignore{
\SK{Also it can be easily proved that 
$4\cdot \sin \frac{A}{2}\sin \frac{B}{2}\sin \frac{C}{2}  = \left( \frac{r}{R}\right)$}
}
 \end{proof}

\section{Technical Properties of the Orthic Patrolling Schedule}
\label{sec: technical properties}

In this section we explore a number of technical properties associated with the orthic patrolling schedule, which will be the cornerstone of our main results. All observations in this section refer to Figure~\ref{fig: technical orthic} which we explain gradually as we present our findings. 

 \begin{figure}[h!]
\centering
  \includegraphics[width=9cm]{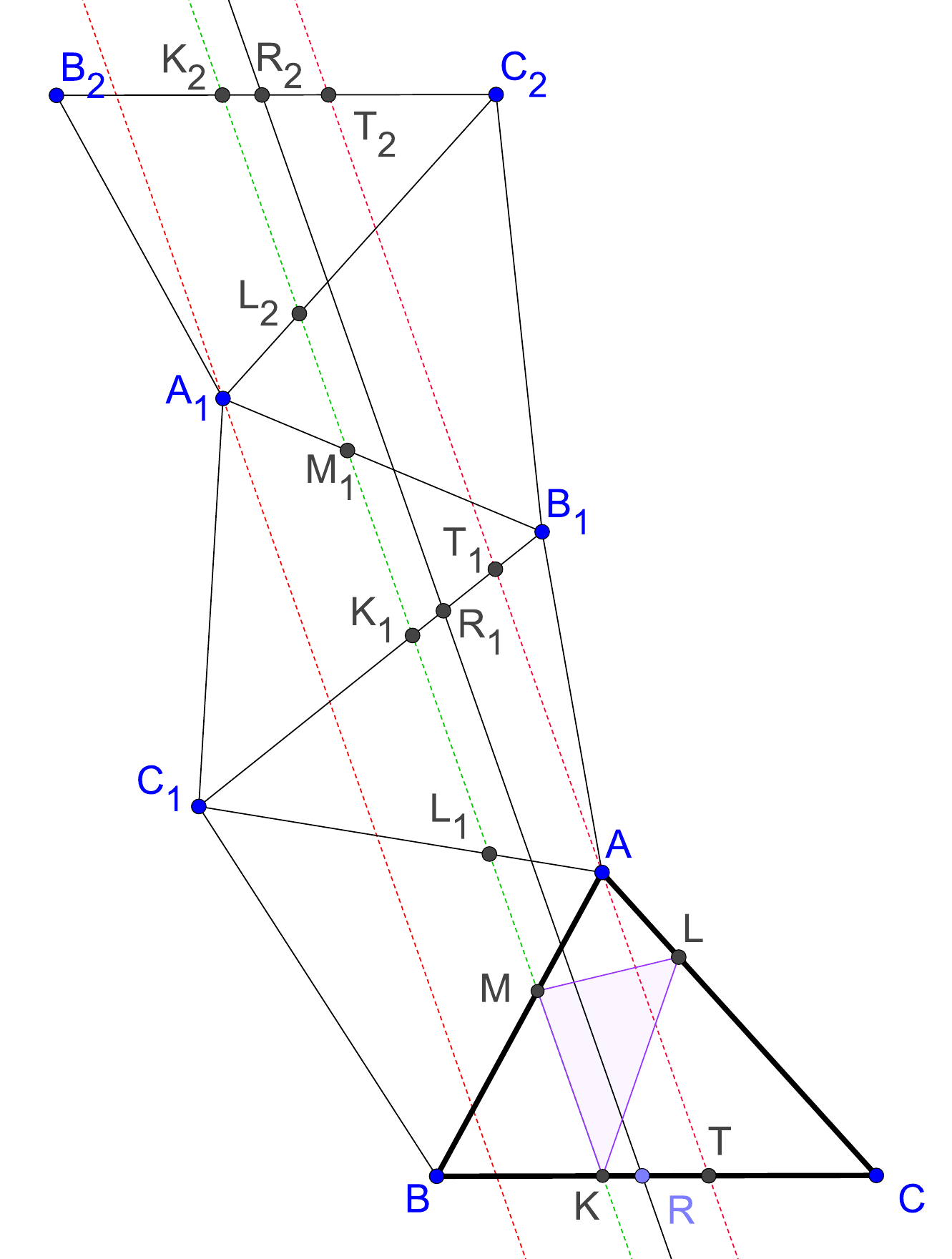}
\caption{The orthic channel (stripe enclosed between the red dotted lines) as it is obtained by 5 triangle reflections.}
\label{fig: technical orthic}
\end{figure}

Our starting point is a triangle $ABC$ with edge lengths satisfying $\alpha \geq \beta \geq \gamma$, where $\alpha, \beta, \gamma$ are the sides opposite to vertices $A, B, C$, respectively. Consequently, the same inequality holds for the corresponding opposite angles, which are labeled according to their respective vertices, i.e. $A,B,C$. We also denote by $K, L, M$ the base points of the altitudes from $A, B,$ and $C$, respectively. It follows that the inscribed triangle $KLM$ is the orthic triangle.

We apply a number of reflections of triangle $ABC$ as follows: 
we obtain 
reflection $C_1$ of $C$ around $AB$, 
reflection $B_1$ of $B$ around $AC_1$, 
reflection $A_1$ of $A$ around $B_1C_1$, 
reflection $C_2$ of $C_1$ around $A_1B_1$, and
reflection $B_2$ of $B_1$ around $A_1C_2$. 
We refer to the resulting triangle construction gadget as the \emph{reflected triangles}.
 
\begin{lemma}
\label{lem: parallel after reflection}
The line passing through $B_2, C_2$ is parallel to the line passing through $BC$. 
\end{lemma}

\begin{proof}
We consider the slope of several line segments relevant to $BC$. We have the following observations pertaining to counterclockwise rotation of line segments about one of their endpoints. 
The rotation of $BC$ about $B$ by angle $2B$ gives segment $BC_1$. 
The rotation of $BC_1$ about $C_1$ by angle $3C$ gives segment $C_1A_1$. 
The rotation of $C_1A_1$ about $A_1$ by angle $3A$ gives segment $A_1B_2$. 
Finally, the rotation of $A_1B_2$ about $B_2$ by angle $B$ gives a point on the line passing through $B_2,C_2$.

It follows that segment $B_2C_2$ follows by repeated rotation by angle
$
2B+3C+3A+B=3(A+B+C)=3\pi.
$ 
Since $3\pi$ is a multiple of $\pi$ we conclude the claim. 
 \end{proof}

Next, we provide an alternative representation of the orthic trajectory.  
For this purpose (see Figure~\ref{fig: technical orthic}), we consider the points $L_1, K_1, M_1, L_2$, which are identified as the intersections of the segment pairs $(BB_1, AC_1)$, $(AA_1, B_1C_1)$, $(C_1C_2, A_1B_1)$, and $(B_1B_2, A_1C_2)$, respectively.  
Let $K_2$ also be the projection of $A_1$ onto $B_2C_2$.  
Thus, the points $L_1$ (on $AC_1$),
$K_1$ (on $B_1C_1$), 
$M_1$ (on $A_1B_1$), 
$L_2$ (on $A_1C_2$), and
$K_2$ (on $B_2C_2$)
are the bases of the corresponding altitudes in the sequence of the reflected triangles, as described in the proof of Lemma~\ref{lem: parallel after reflection}.

\begin{lemma}
\label{lem: orthic line}
The line passing through MK (green dotted line in Figure~\ref{fig: technical orthic}) passes through the points $L_1$, $K_1$, $M_1$, $L_2$, $K_2$. 
\end{lemma}

\begin{proof}
The lemma states, equivalently, that the points $K$, $M$, $L_1$, $K_1$, $M_1$, $L_2$, $K_2$ are collinear.

First we argue that $K,M,L_1$ are collinear using analytic geometry. 
For this, and after proper scaling, we can embed triangle $ABC$ in a Cartesian system with $B=(0,0)$, $A=(x,y)$ and $C=(1,0)$. Since $K$ is the projection of $A$ onto $BC$, we have that $K=(x,0)$. 

Next we calculate the projection $M$ of $C$ onto $AB$. Since $M$ is on $AB$, there exists $t\in \reals$ such that $M=(tx,ty)$. However, $MC$ and $BA$ are by definition perpendicular, hence the corresponding vectors have inner product $0$, from which we can easily derive that $t=x/(x^2+y^2)$. Hence, we conclude that 
$M=\left(\frac{x^2}{x^2+y^2},\frac{xy}{x^2+y^2}\right)$.

Next we derive the reflection $C'$ of $C$ about $AB$. For this, we note that $M$ is the middle point of $CC'$, from which we obtain that $C'=\left(1-\frac{2 y^2}{x^2+y^2},\frac{2 x y}{x^2+y^2}\right)$, which we abbreviate by $(x_0,y_0)$. 

Finally, we calculate the projection $L_1$ of $B=(0,0)$ onto $AC_1$. For this, we note that there exists $r\in \reals$ such that $L_1=r(x_0,y_0)+(1-r)(x,y)$. Requiring that the vector $BL_1$ is orthogonal to the vector $AC'$, we get that $r=\frac{x^2-x x_0+y (y-y_0)}{(x-x_0)^2+(y-y_0)^2}$. 

To conclude, we have that $K=(x,0)$, $M=\left(\frac{x^2}{x^2+y^2},\frac{xy}{x^2+y^2}\right)$, and $L_1=r(x_0,y_0)+(1-r)(x,y)$, with $x_0,y_0,r$ calculated as above. For these values, and with elementary algebraic calculations, we can verify that 
$$
\textsc{Det} \begin{bmatrix}
x & 0 & 1 \\
\frac{x^2}{x^2 + y^2} & \frac{x y}{x^2 + y^2} & 1 \\
r x_0 + (1 - r) x & r y_0 + (1 - r) y & 1
\end{bmatrix}=0,
$$
concluding indeed $K,M,L_1$ are collinear.

An identical argument, shows that $M, L_1, K_1$ are collinear, that $L_1, K_1, M_1$ are collinear, that $K_1$, $M_1$, $L_2$ are collinear, and that $M_1, L_2, K_2$ are collinear, concluding the main claim. 
 \end{proof}
It follows from Lemma~\ref{lem: orthic line} that the orthic trajectory over two cycles of the patrolling schedule can also be described by the line segment $KK_2$. We refer to the line passing through $K$ and $K_2$ as the \emph{orthic line}. Additionally, we showed that all points oof the orthic line lie within the reflected triangles. In particular, the orthic line intersects each reflected triangle at least two edges (and if the line passes through a vertex, we consider it as intersecting the two adjacent edges).  
This observation justifies that the following concept is well-defined.
\begin{definition}
\label{def: orthic channel}
The \emph{orthic channel} is defined by two lines $\ell_1$ and $\ell_2$, which are parallel to the orthic line and at maximum distance from it, with the property that each line intersects at least two edges of each reflected triangle.
\end{definition}
Similar reflection-induced channels were studied by~\cite{Schwartz06,Schwartz09}, while the orthic channel that we use was also observed experimentally by~\cite{lafargue2014localized}. 
Next, we establish its significance formally.

\begin{lemma}
\label{lem: 6periodic within channel}
Any line parallel to the orthic line within the orthic channel gives rise to a cyclic $6$-periodic patrolling schedule with a $2$-gap equal to twice the orthic perimeter, i.e., the perimeter of the orthic triangle.
\end{lemma}

\begin{proof}
Consider an arbitrary line, parallel to the orthic line, that intersects line segments $BC, B_1C_1,B_2C_2$ at points $R,R_1,R_2$ respectively, see Figure~\ref{fig: technical orthic}. We observe that $KK_2$ is parallel to $RR_2$, and by Lemma~\ref{lem: parallel after reflection} we have that $K_2R_2$ is parallel to $KR$. Therefore, $KRR_2K_2$ is a parallelogram with $KR=KR_2$.

We conclude that $R_2$ is the reflection of $R$ using the same reflections that obtained $K_2$ from $K$. But then, it follows $RR_2$ corresponds to cyclic $6$-periodic patrolling schedule of $2$-gap equal to $RR_2=KK_2=KK_1+K_1K_2=2KK_1$, as promised. 
 \end{proof}

Next we identify all cyclic $6$-periodic patrolling schedules of the same $2$-gap value. We note that in the following lemma we make explicit use of that the repeated reflections were done first along the smallest two edges. 
\begin{lemma}
\label{lem: the orthic channel}
The lines identifying the orthic channel are the two lines parallel to the orthic line, one passing through $A$ and one passing through $A_1$. 
\end{lemma}

\begin{proof}
Consider a line parallel to the orthic line passing through $A$, and intersecting $BC$ at $T$ and the line passing through $B_1C_1$ at point $T_1$. We will show that $T_1$ lies in the segment $K_1B_1$. 

First we claim that $KT=K_1T_1$. To see why, recall that $KK_1$ is parallel to $TT_1$. It is enough to show that $KTT_1K_1$ is an isosceles trapezoid. 
Indeed, note that angle $AT_1C_1$ (read counterclockwise) equals angle $KK_1C_1$ (because $TT_1$ is parallel to $KK_1$), and angle $KK_1C_1$ equals angle $BKM$.
Finally, angle $BKM$ equals angle $BTT_1$, because $TT_1$ is parallel to $KK_1$. 
Now $BTT_1=KTT_1$ and $KK_1C_1=TT_1C_1$ and hence  angles $KTT_1$ and $TT_1K_1$ are equal, showing that $KTT_1K_1$ is an isosceles trapezoid as claimed. 

We conclude that in order to show that $T_1$ lies within segment $K_1B_1$ it is enough to show that $KT<KB$. 
At the same time, $K$ is the projection of $A$ onto $BC$, and $K_1$ the projection of $A$ onto $B_1C_1$, hence $K_1$ is the image of $K$ along the two underlying reflections, which implies that $BK=B_1K_1$. Therefore, it is enough to show that the middle point of $BT$ lies within segment $BK$ (or equivalently that $BK\geq BT/2$).
To see why, recall that $AT$ is parallel to $MK$. Moreover, because angle $A$ is at least as large as angle $B$ (that is, our initial reflections where done using the largest edge last), it follows that the base $M$ of altitude $CM$ is closer to $A$ than to $B$. Effectively, this shows that $BM\geq AB/2$.
Therefore, referring to the similar triangles $BKM$ and $BTA$ where segments $MK$ and $AT$ are parallel, we conclude that $BK\geq BT/2$ as wanted. 

Now let the extension of $T T_1$ intersect the line passing through $B_2C_2$ at point $T_2$. 
Since $T_1$ lies within segment $K_1B_1$, and $T_1T_2$ is parallel to $M_1K_1$, it follows that 
$T_1T_2$ intersects $A_1B_1$ at some point between $M_1$ and $B_1$. Observing also that $T_1T_2$ is parallel to $M_1L_2$, we conclude that $T_1T_2$ intersects at some point between $C_2$ and $L_2$, as well as $B_2C_2$ at some point between $K_2$ and $C_2$, which we call $T_2$. 
This shows that the line passing through $AT$ is indeed one of the extreme lines of the orthic channel.

The claim follows by observing that we can repeat the same argument, starting from triangle \( A_1B_2C_2 \) and applying the reverse sequence of reflections that produced the reflected triangles (with \( ABC \) being the final reflected triangle). Note that these reflections would still be applied first with respect to the two smallest edges. Indeed, we can consider a line parallel to the orthic line and passing through \( A_1 \), which, by the same argument, is the other extreme line of the orthic channel.
 \end{proof}

\section{Optimal $2$-Gap Cyclic Schedules}
\label{sec: optimal 2-gap}

The purpose of this section is to prove Theorem~\ref{thm: optimal 2-gap cyclic}, as well as to distinguish the significance of the orthic trajectory with respect to the $2$-gap patrolling problem. 
Indeed, below we prove the theorem by showing that the cyclic $6$-periodic patrolling schedules of Lemma~\ref{lem: 6periodic within channel} are the $2$-gap optimal cyclic schedules of cost twice the perimeter of the orthic triangle. 

Any line parallel to the orthic line within the orthic channel (whose boundaries are given in Lemma~\ref{lem: the orthic channel}) gives rise to a cyclic $6$-periodic schedules that we call \emph{sub-orthic} schedule. We depict such a sub-orthic schedule in Figure~\ref{fig: sub-orthic}.
 \begin{figure}[h!]
\centering
  \includegraphics[width=8cm]{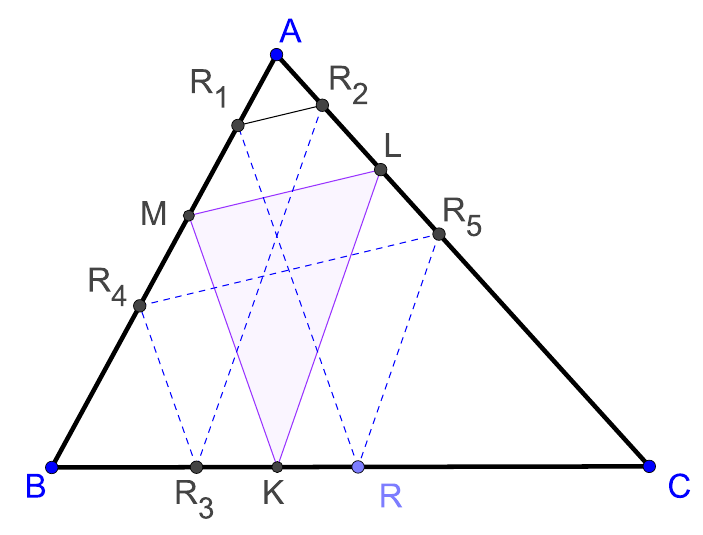}
\caption{
A sub-orthic trajectory example.}
\label{fig: sub-orthic}
\end{figure}
In order to show that any sub-orthic trajectory is $2$-gap optimal, we consider a new patrolling problem on the input triangle $ABC$ with a limited visitation horizon.
In particular, in the \emph{$2k$-limited} patrolling problem, the goal is to find a cyclic trajectory that starts from edge $BC$ (the largest edge), ends after $2k$ visitations of $BC$, and has the minimum total length.
Given triangle $ABC$, we denote by $v_k$ the cost of the optimal solution to the \emph{$2k$-limited} patrolling problem. The following is immediate from our definitions.

\begin{observation}
\label{obs: lower bound to 2gap}
For every $k\geq 1$, the optimal cyclic $2$-gap solution has cost at least $v_k/k$.  
\end{observation}

Now recall that by Lemma~\ref{lem: 6periodic within channel}, any sub-orthic trajectory has $2$-gap equal to twice the orthic triangle. Hence, Theorem~\ref{thm: optimal 2-gap cyclic} is a corollary of the following lemma. 

\begin{lemma}
\label{lem: limit of lower bound}
The value of~ $\lim_{k\rightarrow \infty} v_k/k$ equals twice the perimeter of the orthic triangle. 
\end{lemma}

\begin{proof}
In order to visualize the \emph{$2k$-limited} patrolling problem we apply $k$ times the gadget construction of reflected triangles introduced in Section~\ref{sec: technical properties}, see also Figure~\ref{fig: orthic-horizontal-double} for an example when $k=2$. Next we make the process formal. 

 \begin{figure}[h!]
\centering
  \includegraphics[width=14cm]{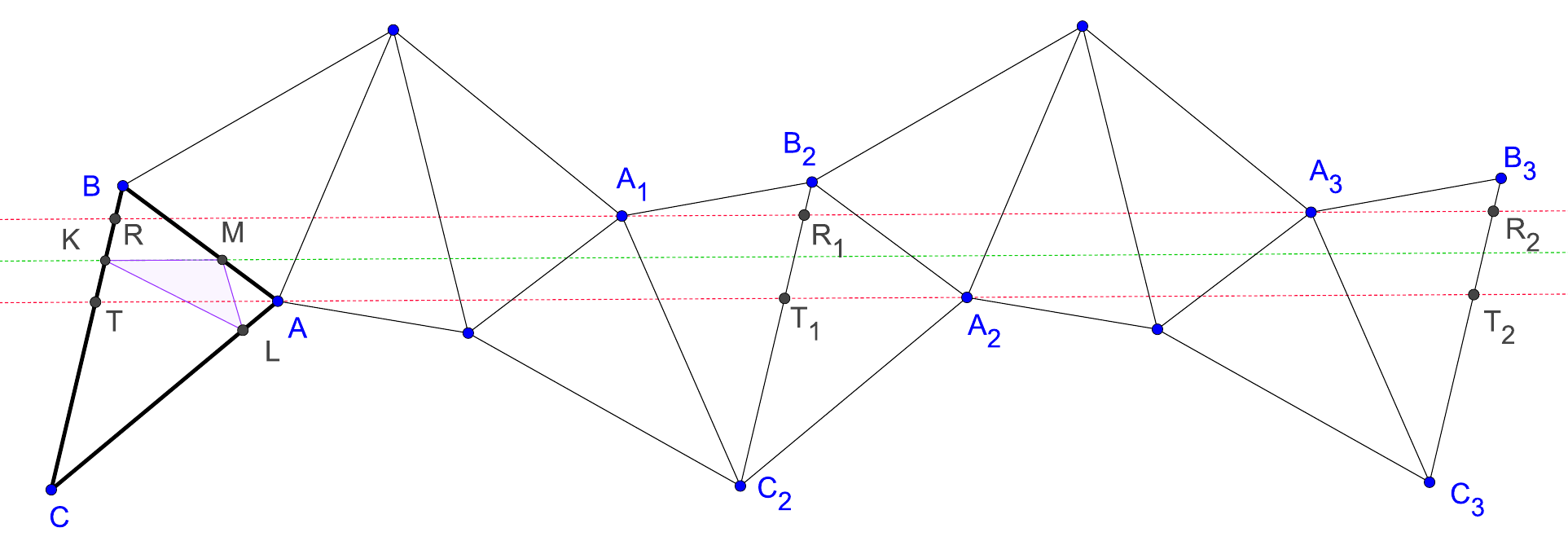}
\caption{
Two applications of reflections.}
\label{fig: orthic-horizontal-double}
\end{figure}

The gadget of the reflected triangles of Figure~\ref{fig: orthic-horizontal-double} defines $B_2C_2$ which is parallel to $BC$. One more reflection of $A_1$ about $B_2C_2$ results into triangle $A_2B_2C_2$ whose edges are piecewise parallel to the edges of $ABC$. 
Starting now with triangle $A_2B_2C_2$, we can implement the same gadget construction of reflected triangles, obtaining this way $B_3C_3$ parallel to $B_2C_2$. 

$k$ applications of the previous construction define a sequence of parallel segments $B_kC_k$. Now consider the orthic channel of $ABC$ identified by the lines passing through $R, A_1$ and $T, A$ (as per Lemma~\ref{lem: the orthic channel}). Also, consider the corresponding points $R_k$ and $T_k$ where these two lines intersect the segments $B_kC_k$.

By the definition of the $2k$-limited patrolling problem, its optimal schedule (with cost $v_k$) is the shortest trajectory that starts from $BC$ and ends at $B_kC_k$. Since the orthic channel stays within all reflected triangles, the optimal solution to the $2k$-limited patrolling problem is the shortest line segment with endpoints within $RT$ and $R_kT_k$. 
Observe that the shortest such segment is the shortest diagonal of the parallelogram $RTT_kR_k$. As $k$ grows, the side $RT$ of these parallelograms remains constant, while $RR_k = TT_k$ tends to infinity. Hence, the ratio of the lengths of the shortest diagonal of $RTT_kR_k$ to $RR_k$ tends to $1$. 
At the same time, $RR_k$ equals $k$ times the $2$-gap of any sub-orthic trajectory and is thus equal to $2k$ times the orthic perimeter.
 \end{proof}

We are now ready to formally prove Theorem~\ref{thm: optimal 2-gap cyclic}.

\begin{proof}[of Theorem~\ref{thm: optimal 2-gap cyclic}]
Observation~\ref{obs: lower bound to 2gap} together with Lemma~\ref{lem: limit of lower bound} imply that the optimal cyclic $2$-gap solution has cost at least twice the perimeter of the orthic triangle. At the same time, every sub-orthic schedule, which is also a cyclic $6$-periodic schedule, has value equal to the perimeter of the orthic triangle, and the claim follows. 
\end{proof}

Note that the orthic trajectory is one among the sub-orthic trajectories, and hence optimal too to the $2$-gap patrolling problem (among cyclic schedules). In the following proposition we show that the orthic trajectory is also an optimal solution to a multi-objective optimization problem. 

\begin{proposition}
\label{lem: orthic is best sub-orthic}
Among all $2$-gap optimal sub-orthic trajectories, the one that minimizes the visitation gap between any two (not necessarily same) edges is the orthic trajectory. 
\end{proposition}

\begin{proof}
Consider an arbitrary sub-orthic trajectory $RR_1R_2R_3R_4R_5R$, see Figure~\ref{fig: sub-orthic}. 
Note that the sub-orthic schedule is made up of segments that are piecewise parallel to the segments of the orthic trajectory, and any of the orthic line segments lies in the middle of any of the two parallel segments of the sub-orthic schedule. 

In particular we have $RR_1, R_3R_4$ are parallel to $MK$, as well as $R_1R_2, R_4R_5$ are parallel to $ML$, and $RR_5, R_2R_3$ are parallel to $KL$. Moreover, 
$MK\leq \max\{ RR_1, R_3R_4 \},
ML\leq \max\{ R_1R_2, R_4R_5 \}$, and 
$KL \leq \max\{ RR_5, R_2R_3\}$. It follows that maximum visitation gap $\max\{MK,ML,KL\}$ between any two edges in the orthic trajectory is at most the maximum visitation gap between any two edges in any sub-orthic trajectory. 
 \end{proof}

\ignore{
 \begin{figure}[h!]
\centering
  \includegraphics[width=8cm]{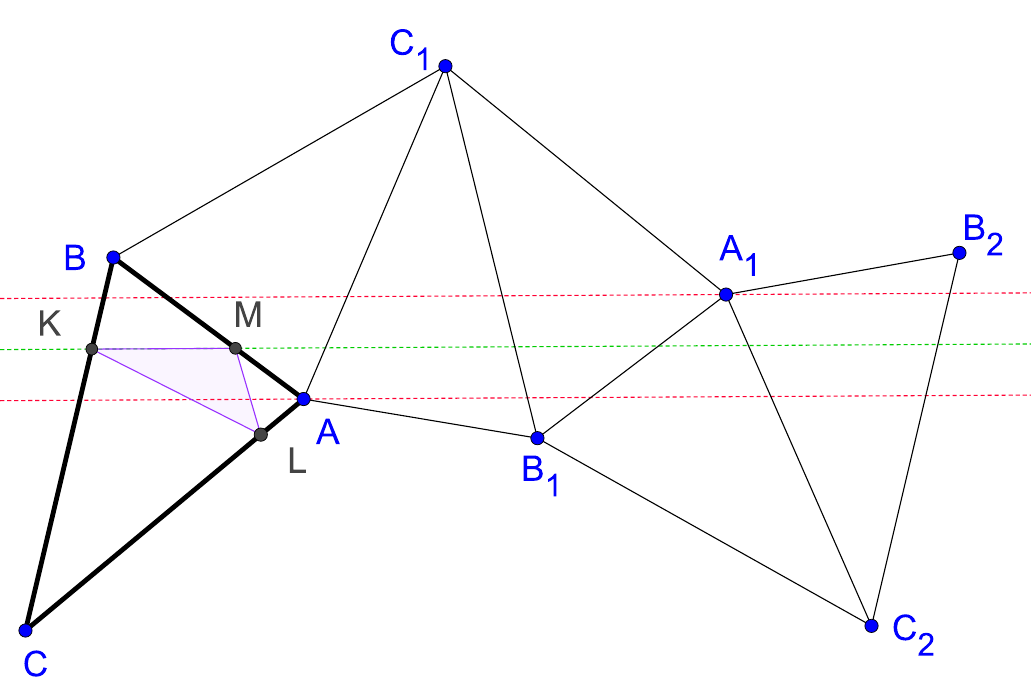}
\caption{
A is largest angle. Reflections are wrt C,B,A,C,B.
Red lines parallel to green (hence to MK edge or orthic triangle.
Red lines pass through images of A. }
\label{fig: 5 reflections}
\end{figure}
}

\section{$1$-Gap Optimal Schedule}
\label{sec: optimal cyclic}

It is immediate from the definitions that half the cost of the $2$-gap optimal patrolling schedule is a lower bound to the cost of the $1$-gap optimal patrolling schedule. By Theorem~\ref{thm: optimal 2-gap cyclic}, the $2$-gap optimal patrolling schedule has cost 2 times the perimeter of the orthic triangle. Hence, the cost of the $1$-gap optimal schedule is at least the perimeter of the orthic triangle. On the other hand, by Theorem~\ref{thm: optimal 1-gap cyclic 3-periodic} we have a patrolling schedule (the orthic trajectory) with $1$-gap equal to the orthic perimeter.
Therefore, we obtain the following immediate corollary. 

\begin{corollary}
\label{cor: optimal 1-gap cyclic}
Restricted to cyclic schedules, the orthic patrolling trajectory of a triangle is $1$-gap optimal. 
\end{corollary}

The purpose of this section is to prove Theorem~\ref{thm: optimal 1-gap cyclic}, that is to strengthen the statement of Corollary~\ref{cor: optimal 1-gap cyclic} by showing that the optimal $1$-gap schedule is actually cyclic. We do so by showing how to modify an arbitrary schedule into a cyclic schedule, without increasing its $1$-gap. Effectively, the  next lemma implies Theorem~\ref{thm: optimal 1-gap cyclic}.

\begin{lemma}
\label{lem: there is optimal 1-gap that is cyclic}
There is a $1$-gap optimal schedule that is cyclic. 
\end{lemma}

\begin{proof}
Consider an arbitrary schedule $S=\{s_i\}_{i\geq 0}$ that is not cyclic. We show how to construct a new schedule that is cyclic and 3-periodic, without increasing its $1$-gap. Indeed, since $S$ is not cyclic, and after renaming edges, there are two consecutive visitations of edge $\alpha$ so that both edges $\beta,\gamma$ are visited in between, with at least one of them being visited more than once. In other words, for some $k,\ell\in \naturals$, $\ell\geq 4$ we have that $e(s_k)=e(s_{k+\ell})=\alpha$, $e(s_{k+1})=e(s_{k+3})=\beta$ and $e(s_{k+2})=\gamma$. 

Then, we see that for the $1$-gap $G$ of $S$, we have that 
\begin{align*}
G 
= &  \max_{\delta \in E} G(\delta) 
\geq  G(\alpha) 
\geq \sum_{i=0}^{\ell-1} s_{k+i}s_{k+i+1} \\
\geq  &
s_ks_{k+1} 
+ s_{k+1}s_{k+2}
+ s_{k+2}s_{k+3}
+ s_{k+3}s_{k+\ell} \\
\geq& 2 \min\{
s_ks_{k+1} 
+ s_{k+1}s_{k+2}, 
s_{k+2}s_{k+3}
+ s_{k+3}s_{k+\ell}
  \},
\end{align*}
where the second to last inequality is due to the triangle inequality. 

Now we consider two different cyclic and $3$-periodic schedules, $S',S''$, with $1$-gap values $G',G''$, respectively, and we show that $\min\{G',G''\} \leq G$. The two schedules are the following. 
\begin{align*}
S' & = s_k,s_{k+1},s_{k+2},  s_k,s_{k+1},s_{k+2},s_k,s_{k+1},s_{k+2},\ldots \\
S'' & = s_{k+2},s_{k+3},s_{k+\ell},s_{k+2},s_{k+3},s_{k+\ell},s_{k+2},s_{k+3},s_{k+\ell},
\ldots
\end{align*}
Since both $S',S''$ are cyclic and periodic, we have that 
$G'=G'(\alpha)=G'(\beta)=G'(\gamma)$ and 
$G''=G''(\alpha)=G''(\beta)=G''(\gamma)$. In particular, using the triangle inequality again, we have 
\begin{align*}
G'
&=
s_ks_{k+1}+ 
s_{k+1}s_{k+2}
+ s_{k+2}s_k \leq 2(s_ks_{k+1}+ s_{k+1}s_{k+2}) \\
G''
&=
s_{k+2}s_{k+3}
+
s_{k+3}s_{k+\ell}
+s_{k+\ell}s_{k+2}
\leq 2(s_{k+2}s_{k+3}
+
s_{k+3}s_{k+\ell}).
\end{align*}
But then, 
$\min\{G',G''\} \leq 
2 \min\{
s_ks_{k+1} 
+ s_{k+1}s_{k+2}, 
s_{k+2}s_{k+3}
+ s_{k+3}s_{k+\ell}
  \}
  \leq G$, as wanted. 
 \end{proof}

\section{The Greedy Cyclic Schedule}
\label{sec: greedy}

In this section, we prove Theorem~\ref{thm: greedy}, that is, we describe a patrolling schedule that converges to a $3$-periodic cyclic schedule whose $1$-gap differs from the $1$-gap optimal cyclic schedule by a factor $\gamma \in [1,1.20711]$. It will follow from our analysis that our greedy algorithm is nearly optimal for any acute triangle with one sufficiently small angle, and it performs worst compared to the optimal solution when the given triangle is a right isosceles.

We proceed by the description of a greedy patrolling schedule. We assume that the patroller can remember the current and previously visited edges (not necessarily their points), as well as that it can compute (move along) the projection of its current position to any other edge.
Formally, we label the three edges $BC,AB,AC$ as $0,1,2$, respectively. The patrolling schedule starts from an arbitrary point $p_0$ on $BC$. For each $i\geq 1$, the patroller moves to point $p_i$, which is the projection of $p_{i-1}$ onto edge $i \bmod 3$. Referring to triangle $ABC$ as in Figure~\ref{fig: greedy}, we note that the patrolling schedule induces a clockwise cyclic visitation of the given triangle. An immediate corollary of our results will imply that also the corresponding counterclockwise cyclic visitation induces the same $1$-gap. 

\begin{minipage}{\linewidth}
      \centering
      \begin{minipage}{0.35\linewidth}
\begin{figure}[H]
\centering
  \includegraphics[width=6.1cm]{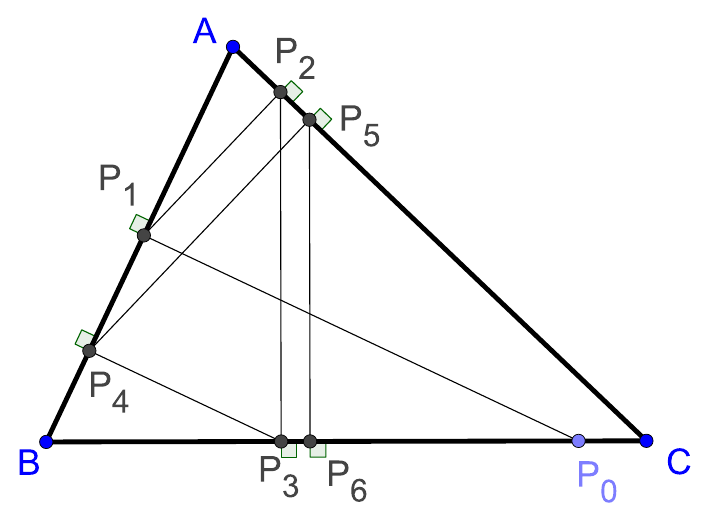}
\caption{
Six iterations of the greedy patrolling schedule that starts from point $p_0$ of edge $BC$.
}
\label{fig: greedy}
\end{figure}
      \end{minipage}
      \hspace{0.15\linewidth}
      \begin{minipage}{0.42\linewidth}
\begin{figure}[H]
\centering
  \includegraphics[width=6.1cm]{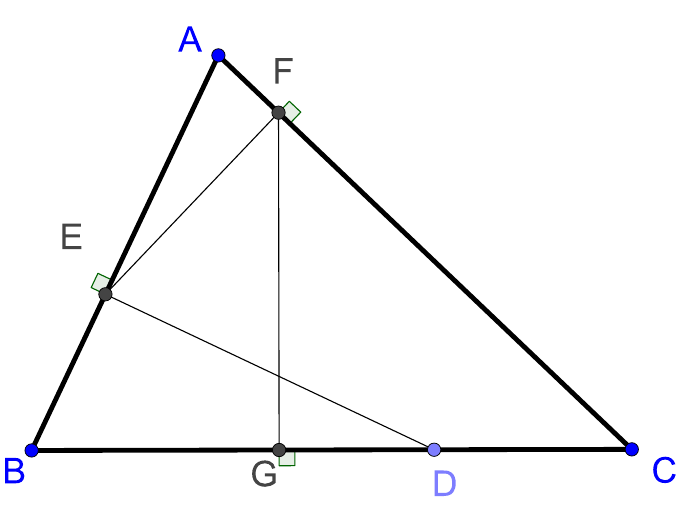}
\caption{One iteration of the greedy patrolling schedule, starting from point $D$.}
\label{fig: 1iteration}
\end{figure}
      \end{minipage}
  \end{minipage}

\begin{lemma}
\label{lem: greedy converges}
For any acute triangle $ABC$ and any initial starting point, the greedy algorithm converges to a cyclic $3$-periodic schedule that has $1$-gap 
\begin{equation}
\label{equa: greedy cost}
p  \cdot \frac{\sin(A)\sin(B)\sin(C)}{1+\cos(A)\cos(B)\cos(C)}
\end{equation}
where $p$ is the perimeter of triangle $ABC$.
\end{lemma}

\begin{proof}
Consider an arbitrary iteration of the greedy algorithm and a point $D$ on $BC$, see Figure~\ref{fig: 1iteration}. After 3 consecutive steps, the patroller has moved to the projection $E$ of $D$ onto $AB$, its projection $F$ on $AC$ and to its projection $G$ back on $BC$. To simplify calculations, assume also that $BC$ has length $1$. Below, we derive a relation between $BG$ and $BD$. 

First we note that $AF=\cos(A) AE = \cos(A) (\gamma-BE) = \cos(A)(\gamma - \cos(B) BD)$. Then, we use the derived formula for $AF$ to calculate 
\begin{align*}
BG &= 1-CG 
		=1 - \cos(C) CF 
		=1 - \cos(C) (\beta - AF) 		
		=1 - \cos(C) \left(\beta - \cos(A)(\gamma - \cos(B) BD) \right). 
\end{align*}
It follows that 
$BG = c - \cos(A) \cos(B) \cos(C) \, BD,$ 
where the constant $c = 1 - \cos(C) \beta + \cos(A) \cos(C) \gamma$ is independent of the points $G$ and $D$. 

In the greedy patrolling schedule, edge $BC$ is visited infinitely often. We denote by $d_i$ the distance between the point on $BC$ at the $i$-th visitation and point $B$. In particular, we have $d_1 = BD$ and $d_2 = BG$. 
Therefore, we have shown that 
$d_{2} = c - \cos(A) \cos(B) \cos(C) \, d_1$, 
and, by induction, we conclude that 
$d_{i+1} = c - \cos(A) \cos(B) \cos(C) \, d_i$ 
for all $i \geq 1$.

The latter is a non-homogeneous linear recurrence relation of degree $1$ of the form 
$d_{i+1} = c - x \, d_i$, 
where $x = \cos(A)\cos(B)\cos(C)$, and note that $|x| < 1$.
One particular solution to this recurrence is $c / (1 + x)$. Moreover, the closed-form solution to the corresponding homogeneous recurrence is $(-x)^{n-1} d_1$. 
It follows that $\lim_{i \rightarrow \infty} d_i$ exists, and its value is obtained when, in the previous argument, the points $D$ and $G$ coincide (see Figure~\ref{fig: limit}).

\begin{figure}[H]
\centering
  \includegraphics[width=8cm]{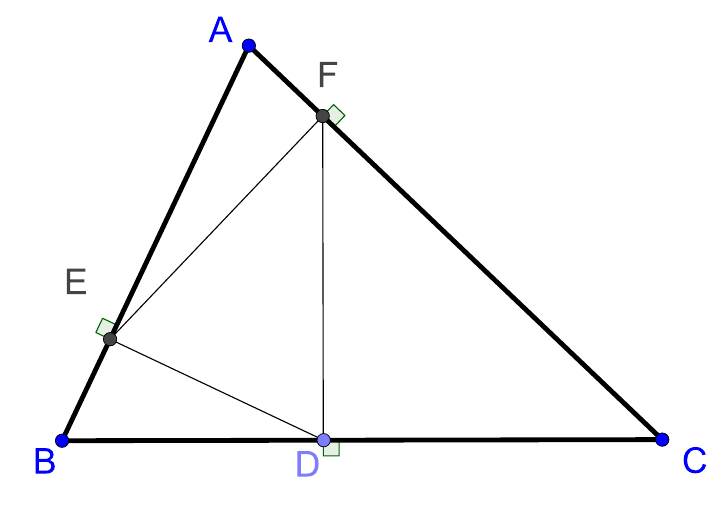}
\caption{
The limiting cyclic $3$-periodic trajectory of the (clockwise) greedy algorithm}
\label{fig: limit}
\end{figure}

We proved that inscribed triangle $DEF$ is the limiting patrolling schedule of the greedy algorithm, which is indeed a cyclic $3$-periodic schedule. Next we calculate its cost. 
To this end, we claim that triangles $DEF$ and $ABC$ are similar. By denoting by $F,E,G$ the angles of the inscribed triangle, and looking at right triangle $FD$ we have 
$
F=\pi - \pi/2 - (\pi-C-\pi/2)=C
$.
Similarly we obtain that angles $D,B$ are equal, and angles $E,A$ are equal. 

Finally we compute the similarity ratio $k<1$ of triangles $DEF,ABC$. We have that
$$
\alpha=BD+DC=\frac{ED}{\sin(B)}+\frac{DF}{\tan(C)} 
=
\frac{k \gamma}{\sin(B)} + \frac{k \alpha}{\tan(C)}
=
\frac{k \alpha \sin(C)}{\sin(B)} + \frac{k \alpha}{\tan(C)},
$$ 
where the last equality follows from the sin Law in triangle $ABC$. But then, solving for $k$ and simplifying the trigonometric expressions yields 
$
k=\frac{\sin(A)\sin(B)\sin(C)}{1+\cos(A)\cos(B)\cos(C)}.
$
It follows that the $1$-gap of the induced patrolling schedule is equal to the perimeter of triangle $DEF$ which equals $k$ times the perimeter of $ABC$ as claimed. 
 \end{proof}

We are now ready to prove Theorem~\ref{thm: greedy}. 
An immediate corollary of Lemma~\ref{lem: greedy converges} is that the (limiting) cost of the greedy algorithm is the same also for the corresponding counter-clock wise trajectory. Moreover, the ratio between its cost and the optimal $1$-gap, as per Lemma~\ref{lem: cost of orthic}, is given by 
\begin{align*}
& \frac{\sin(A)\sin(B)\sin(C)}{2(1+\cos(A)\cos(B)\cos(C))} 
\left(
\frac1{\sin(B)\sin(C)}
+
\frac1{\sin(A)\sin(C)}
+
\frac1{\sin(A)\sin(B)}
\right) \\
& =\frac{\sin (A)+\sin (B)+\sin (C)}{2(1+\cos(A)\cos(B)\cos(C))}. 
\end{align*}
Call the latter expression $f(A, B, C)$. We aim to find the extreme values of $f(A, B, C)$, subject to the constraints $A + B + C = \pi$ and $0 \leq A, B, C \leq \pi/2$, resulting in a nonlinear program. Due to the symmetry of the objective function and by the KKT conditions, if there exists an extreme point where none of the inequality constraints are satisfied tightly, we must have $A = B = C = \pi/3$. In this case, 
$f(\pi/3, \pi/3, \pi/3) = 2\sqrt{3}/3$.

Next, we investigate the values of the objective function at the boundary of its feasible region. Without loss of generality, we may assume that $A \leq B \leq C$. The constraint $C \leq \pi/2$ is either tight at the optimizer or not. 

If it is tight, the objective simplifies to 
$f(A, \pi/2 - A, \pi/2) = (1 + \cos(A) + \sin(A)) / 2$,
where $A \leq \pi/3$. Elementary calculus shows that this expression is maximized when $A = \pi/4$, yielding
$f(\pi/4, \pi/4, \pi/2) = (1 + \sqrt{2}) / 2$,
and minimized when $A = 0$, for which 
$f(0, \pi/2, \pi/2) = 1$.

Lastly, we examine the case where the constraint $C \leq \pi/2$ is not active. In this case, the constraint $A \geq 0$ must be active; otherwise, none of the inequality constraints are active, a scenario we have already considered. Here, the objective simplifies to 
$f(0, \pi - C, C) = 1 / \sin(C)$,
where $\pi/3 \leq C < \pi/2$. The infimum of this expression equals $1$, attained as $C \to \pi/2$, while its maximum equals 
$2\sqrt{3}/3$
attained at $C = \pi/3$.

Overall, we have shown that the ratio between the cost of the greedy algorithm and the optimal $1$-gap value (over all non-obtuse triangles) is maximized when one of the angles $A, B, C$ is a right angle, and the other two are equal, that is, in the case of the right isosceles triangle. In this case, the ratio becomes 
$\frac{1}{2} (1 + \sqrt{2})$.
We also showed that the ratio tends to $1$ as any of the angles tends to $0$ (causing the other two to approach $\pi/2$), and for the equilateral triangle, the ratio equals 
$2\sqrt{3}/3$.

\section{Discussion}

In this work we demonstrated the connection between billiard-type trajectories and optimal patrolling schedules in combinatorial optimization. Specifically, we introduced and solved the problem of patrolling the edges of an acute triangle using a unit-speed agent with the goal of minimizing the maximum 1-gap and 2-gap idle time of any edge. We show that billiard-type trajectories are optimal solution to these combinatorial patrolling problems. 

Our findings point to several future directions. One natural extension of our work is to generalize the patrolling problem to arbitrary polygons with one or more agents. Moreover, the introduction of the novel 2-gap patrolling problem suggests the investigation of optimal solutions for more complex frequency requirements or time restrictions, especially with the presence of multiple patrolling agents or multiple objects to be patrolled. 
In that direction, it would be interesting to examine how our results extend to patrolling 3 or more arbitrary line segments on the plane, as subsets of the edges of convex polygones with one or more agents.


\nocite{*}
\bibliographystyle{abbrvnat}
\bibliography{newbib-updated}
\label{sec:biblio}

\end{document}